\documentclass[11pt]{amsart}

\usepackage[T1]{fontenc}
\usepackage[utf8]{inputenc}
\usepackage{lmodern}
\usepackage{microtype}
\usepackage{amsmath,amssymb,amsthm,mathtools,mathrsfs}
\usepackage{geometry}
\usepackage{hyperref}
\usepackage[nameinlink,noabbrev]{cleveref}
\usepackage{enumitem}

\hypersetup{hidelinks}

\numberwithin{equation}{section}

\theoremstyle{plain}
\newtheorem{theorem}{Theorem}[section]
\newtheorem{proposition}[theorem]{Proposition}
\newtheorem{lemma}[theorem]{Lemma}
\newtheorem{corollary}[theorem]{Corollary}

\theoremstyle{definition}
\newtheorem{definition}[theorem]{Definition}
\newtheorem{remark}[theorem]{Remark}

\newcommand{\cC}{\mathcal C}
\newcommand{\cE}{\mathcal E}
\newcommand{\cH}{\mathcal H}
\newcommand{\cK}{\mathcal K}
\newcommand{\cL}{\mathcal L}
\newcommand{\cN}{\mathcal N}
\newcommand{\cS}{\mathcal S}
\newcommand{\cU}{\mathscr U}
\newcommand{\cW}{\mathcal W}
\newcommand{\Ghigh}{\mathscr G_{>}}
\newcommand{\Gad}{\mathscr G_{\mathrm{ad}}}
\newcommand{\grw}{\operatorname{gr}_{w}}
\newcommand{\Sky}{\operatorname{Sky}}
\newcommand{\Sp}{\operatorname{Sp}}
\newcommand{\PW}{\mathsf{PW}}

\title{Sachs Equations and Plane Waves VI:\\
Soldered Penrose Limits}

\author{Jonathan Holland}
\author{George Sparling}
\address{University of Pittsburgh\\
Department of Mathematics\\
301 Thackeray Hall\\ Pittsburgh, PA 15260}
\date{August 9, 2026}

\begin{document}

\begin{abstract}
Let $(M,g)$ be a Lorentzian manifold whose local space $\cN$ of
unparametrized null geodesics is smooth.  We show that its Penrose limits
assemble into a smooth bundle of plane-wave germs with a canonical
gauge-theoretic soldering to spacetime.  The incidence correspondence
$\cU=\{(x,[\gamma]):x\in\gamma\}$ identifies the contact space modulo the
tangent space to the sky of $x$ with the screen space at $x$, and identifies
the contact line with the weight-two null quotient.  These evaluation maps
solder the associated-graded Penrose model.  A first jet of contact scale fixes
the affine scale and a Reeb representative of the weight-two direction.  Above
the resulting pullback of $J^1\cS$, full
neighborhood realizations form a torsor under the group $\Ghigh$ of
weighted diffeomorphism germs with identity principal part. We also identify the sky-parabolic reduction and its
further transverse-Lagrangian reduction to Rosen gauges.  The reduction fails
precisely where the chosen Lagrangian meets the moving sky, which is the Rosen
caustic locus.  The trace of the Grassmannian Schwarzian of the sky curve
defines a conformally natural projective structure on each null geodesic, and
its trace-free part is the Weyl tidal profile of the Penrose germ.
\end{abstract}

\maketitle

\section{Introduction}
\label{sec:introduction}

The Penrose limit assigns a plane wave to a null geodesic in a Lorentzian
spacetime.  In adapted coordinates $(u,x^1,\ldots,x^n,v)$ it is obtained from
the anisotropic dilation
\begin{equation}\label{eq:intro-dilation}
  \delta_\epsilon(u,x,v)=(u,\epsilon x,\epsilon^2v)
\end{equation}
and the limit of $\epsilon^{-2}\delta_\epsilon^*g$ as
$\epsilon\to0$ \cite{Penrose1976,Guven2000}.  Null Fermi coordinates make the
curvature profile of the resulting plane wave manifestly covariant
\cite{Blau2006}.  What remains less transparent in the coordinate definition is
the sense in which this limit is still attached, or \emph{soldered}, to the
original spacetime.

The null-Fermi expansion of Blau, Frank, and Weiss identifies the leading
profile and its higher corrections with curvature and its covariant derivatives
along the ray \cite{Blau2006}.  Here we ask how the leading germs assemble as
the ray varies and how much of an adapted realization survives the weighted
blow-up.  The weighted associated-graded metric is canonical.  A first jet of
contact scale selects the affine scale and the degree-two Reeb representative.
Full neighborhood realizations with that principal part form a higher-order
gauge orbit.  We use \emph{embedding germ} for such a weighted tubular
realization, which compares the ambient metric with the plane-wave metric as
its leading weighted symbol.

The main result, \Cref{thm:canonical-soldering}, states this globally.  Let
$\cN$ be a local smooth space of
unparametrized null geodesics and let
\[
 \cU=\{(x,[\gamma])\in M\times\cN:x\in\gamma\}
\]
be the incidence correspondence.  Write $\cS\to\cN$ for the bundle of positive
contact scales.  On
\[
 \widehat{\cU}=\cU\times_{\cN}J^1\cS
\]
we construct a bundle of Penrose plane-wave germs.  Its associated-graded
configuration bundle is canonically identified with the weighted normal bundle
of the corresponding null geodesics in $M$.  Full weighted realizations of that
identification form a torsor under $\Ghigh$, the group of admissible
diffeomorphism germs whose weighted principal part is the identity.  In this
way the bundle of Penrose-limit germs is canonically soldered into spacetime.

The construction uses the contact geometry of $\cN$.  At an incident pair
$q=(x,[\gamma])$, the tangent space to the sky of $x$ is a Lagrangian subspace
$\cH_q$ of the contact hyperplane $\cC_{[\gamma]}$.  Evaluation of Jacobi
fields at $x$ gives canonical isomorphisms
\begin{equation}\label{eq:intro-evaluation}
 \cC_{[\gamma]}/\cH_q\cong
 \dot\gamma_x^\perp/\langle\dot\gamma_x\rangle,
 \qquad
 T_{[\gamma]}\cN/\cC_{[\gamma]}\cong
 T_xM/\dot\gamma_x^\perp.
\end{equation}
These two maps are the soldering: the first identifies the weight-one screen
directions and the second identifies the weight-two line.  A contact-scale jet
supplies the scale and Reeb splitting needed to choose representatives of
these quotient maps.

Rosen coordinates are described by Lagrangian complements.  A Lagrangian
$\Lambda\subset\cC_{[\gamma]}$ gives a Rosen gauge only where it is transverse
to the moving sky Lagrangian $\cH_q$.  A Rosen caustic occurs when this
transversality is lost.  The stabilizer of $\cH_q$ is a symplectic parabolic;
its contact lift is the finite-dimensional structure group of the principal
gauge which survives the blow-up.

The sky curve also carries a conformal refinement.  Its Grassmannian
Schwarzian transforms under reparametrization by an inhomogeneous scalar
Schwarzian term.  The normalized trace is therefore a one-dimensional
projective connection along $\gamma$, and the trace-free part is an
endomorphism-valued quadratic differential representing the Weyl tidal
curvature.  Thus the conformal null geometry selects a projective class of
parameters.  Affine normalization remains part of the contact-scale data.
This is proved in \Cref{thm:projective-refinement}, using the projective
differential geometry developed in \cite{HollandSparlingIV} and the conformal
plane-wave transformation law of \cite{HollandSparlingII}.

The transverse Jacobi equation, its Hamiltonian curve in the Lagrangian
Grassmannian, and the equivalence between Brinkmann and Rosen data were developed
in the preceding papers of this series
\cite{HollandSparlingI,HollandSparlingII,HollandSparlingIII,
HollandSparlingIV,HollandSparlingV}.  The
present paper assembles their Penrose germs into a bundle and determines its
soldering and residual gauge.

The paper is organized as follows.  \Cref{sec:weighted-germ} recalls the
transverse Jacobi system and identifies the plane wave with the weighted metric
symbol of the ambient germ.  \Cref{sec:weighted-gauge} separates the surviving
homogeneous gauge from the higher-order realization gauge.
\Cref{sec:contact-incidence} derives the contact and Heisenberg models from the
null cone and proves the incidence evaluation theorem.  The unpolarized,
sky-parabolic, and Rosen bundles are constructed in
\Cref{sec:soldered-bundle}, where the canonical soldering theorem is proved.
\Cref{sec:rosen-legendrian} treats Rosen caustics and the additional geometry
supplied by Legendrian families.  Finally,
\Cref{sec:conformal-naturality} proves the conformally natural projective
refinement of the longitudinal parameter.

\section{The weighted Penrose germ}
\label{sec:weighted-germ}

Let $(M,g)$ have dimension $n+2$ and Lorentz signature
$(-,+,\ldots,+)$, and let $\gamma:I\to M$ be an affinely parametrized null
geodesic with tangent $k=\dot\gamma$.  Along $\gamma$ there is the null flag
\begin{equation}\label{eq:null-flag}
 \langle k\rangle\subset k^\perp\subset TM|_\gamma.
\end{equation}
Its screen and complementary null quotients are
\begin{equation}\label{eq:screen-and-line}
 E=k^\perp/\langle k\rangle,
 \qquad N_\gamma=TM|_\gamma/k^\perp.
\end{equation}
The metric induces a positive-definite screen metric $h$ on $E$ and a
nondegenerate pairing $\langle k\rangle\otimes N_\gamma\to\mathbb R$.  We
assign weight one to $E$ and weight two to $N_\gamma$ and write
\begin{equation}\label{eq:weighted-normal}
 W_\gamma=E[1]\oplus N_\gamma[2]\longrightarrow\gamma.
\end{equation}
The zero section carries the unscaled longitudinal variable $u$.  The intrinsic
dilation on the fibers is
\begin{equation}\label{eq:weighted-dilation}
 \delta_\epsilon(x,v)=(\epsilon x,\epsilon^2v).
\end{equation}

The full associated graded of the null flag is
\begin{equation}\label{eq:full-associated-graded}
 \grw(TM|_\gamma)
 =\langle k\rangle[0]\oplus E[1]\oplus N_\gamma[2].
\end{equation}
It carries no preferred splitting of $TM|_\gamma$.  A screen subbundle
$E^\sharp\subset k^\perp$ and a complementary
null field $\ell$ with $g(k,\ell)=1$ give a realization
\[
 TM|_\gamma=\langle k\rangle\oplus E^\sharp\oplus\langle\ell\rangle,
\]
in which both $E^\sharp$ and $\ell$ depend on choices.  The dilation on
\eqref{eq:weighted-normal} is defined independently of them.

\subsection{Adapted coordinates and the classical scaling}

Penrose-adapted coordinates may be constructed from a local foliation by null
hypersurfaces.  Let $v$ label the hypersurfaces, choose $u$ affine along their
null generators, and propagate $n$ transverse labels $x^i$ along those
generators.  After normalizing the mixed null term, the distinguished geodesic
is $\gamma=\{x=v=0\}$ and the metric has the form below.  The construction
depends on the chosen congruence, while its weighted leading metric is
independent of that choice.

In Penrose-adapted coordinates the metric may be written
\begin{equation}\label{eq:adapted-metric}
 g=2\,du\,dv+A\,dv^2+2B_i\,dv\,dx^i+C_{ij}\,dx^i dx^j,
\end{equation}
where $\gamma=\{x=v=0\}$ and $u$ is affine on $\gamma$.  Direct substitution,
with no singular coefficients, gives
\begin{align}\label{eq:scaled-metric}
 \epsilon^{-2}\delta_\epsilon^*g
  ={}&2\,du\,dv
  +\epsilon^2 A(u,\epsilon x,\epsilon^2v)\,dv^2 \notag\\
  &+2\epsilon B_i(u,\epsilon x,\epsilon^2v)\,dv\,dx^i
  +C_{ij}(u,\epsilon x,\epsilon^2v)\,dx^i dx^j.
\end{align}
Consequently the limit is the Rosen metric
\begin{equation}\label{eq:rosen-limit}
 g_R=2\,du\,dv+C_{ij}(u,0,0)\,dx^i dx^j.
\end{equation}
This presentation includes the choice of the null congruence used to construct
the coordinates.  The isometry class of the plane wave in
\eqref{eq:rosen-limit} is independent of that congruence.

\subsection{Transverse Jacobi geometry}
\label{sec:transverse-jacobi}

The screen bundle $E$ is preserved by the connection induced from
$\nabla_k$.  The curvature endomorphism
\begin{equation}\label{eq:curvature-endomorphism}
 \mathcal R_\gamma:E\longrightarrow E,
 \qquad [X]\longmapsto[R(X,k)k]
\end{equation}
is well defined and self-adjoint.  A transverse Jacobi field is a section of
$E$ satisfying
\begin{equation}\label{eq:jacobi}
 \nabla_k^2J+\mathcal R_\gamma J=0.
\end{equation}
Write $\mathcal J_\gamma^\perp$ for the $2n$-dimensional solution space of
\eqref{eq:jacobi}.  It carries the constant Wronskian form
\begin{equation}\label{eq:wronskian}
 \omega(J_1,J_2)
 =h(\nabla_kJ_1,J_2)-h(J_1,\nabla_kJ_2).
\end{equation}
Constancy follows by differentiating and using the self-adjointness of
$\mathcal R_\gamma$; nondegeneracy follows from the initial-value theorem.
Thus $\mathcal J_\gamma^\perp$ is a symplectic vector space.

For $u_0\in I$, evaluation of Cauchy data gives a symplectic isomorphism
\begin{equation}\label{eq:jacobi-cauchy-data}
 \mathcal J_\gamma^\perp\xrightarrow{\sim}
 E_{u_0}\oplus E_{u_0},
 \qquad J\longmapsto(J(u_0),\nabla_kJ(u_0)).
\end{equation}
The subspace
\begin{equation}\label{eq:vanishing-lagrangian}
 H_\gamma(u_0)=\{J\in\mathcal J_\gamma^\perp:J(u_0)=0\}
\end{equation}
is therefore Lagrangian.  As $u$ varies, these subspaces form the Hamiltonian
curve
\begin{equation}\label{eq:hamiltonian-curve}
 H_\gamma:I\longrightarrow\operatorname{LG}
 (\mathcal J_\gamma^\perp).
\end{equation}
The curvature endomorphism, the symplectic Jacobi system, and this Lagrangian
curve are equivalent descriptions of the transverse plane-wave data
\cite{Arnold1989,Duistermaat2010,HollandSparlingI,HollandSparlingII,
HollandSparlingIV}.

A Rosen presentation is obtained from a fixed Lagrangian
$\Lambda\subset\mathcal J_\gamma^\perp$ on an interval where
$\Lambda\cap H_\gamma(u)=0$.  Evaluation then identifies $\Lambda$ with $E_u$.
If $J_1,\ldots,J_n$ is a basis of $\Lambda$, the matrix
\begin{equation}\label{eq:rosen-jacobi-matrix}
 h_{ij}(u)=h(J_i(u),J_j(u))
\end{equation}
gives the Rosen metric $2\,du\,dv+h_{ij}(u)\,dy^i dy^j$.  The condition that
$\Lambda$ be Lagrangian is the vanishing of the Wronskians and hence the
absence of the unwanted skew cross term.  Loss of transversality to
$H_\gamma(u)$ produces a Rosen caustic.

\subsection{The curvature symbol}

Let $e_1,\ldots,e_n$ be a parallel orthonormal frame of $E$.  With the curvature
convention $R(X,Y)=\nabla_X\nabla_Y-\nabla_Y\nabla_X-\nabla_{[X,Y]}$, set
\begin{equation}\label{eq:profile}
 A_{ij}(u)=-g(R(e_i,k)k,e_j)|_{\gamma(u)}.
\end{equation}
The corresponding Brinkmann metric on $W_\gamma$ is
\begin{equation}\label{eq:brinkmann-limit}
 g_{\PW}=2\,du\,dv+\delta_{ij}\,dx^i dx^j
              +A_{ij}(u)x^ix^j\,du^2.
\end{equation}
Changing the parallel orthonormal frame changes \eqref{eq:brinkmann-limit} by
a constant orthogonal gauge.  Intrinsically, its profile is the negative of
the self-adjoint curvature endomorphism
\eqref{eq:curvature-endomorphism}.

\begin{proposition}[Penrose symbol]\label{prop:penrose-symbol}
The weighted degree-two symbol of the ambient metric germ along $\gamma$ is the
plane-wave germ \eqref{eq:brinkmann-limit}.  It depends only on the null flag,
the affine scale of $k$, the induced screen connection and metric, and the
curvature endomorphism \eqref{eq:curvature-endomorphism}.  In particular it is
independent of a Rosen congruence or a null Fermi frame.
\end{proposition}

\begin{proof}
Choose null Fermi coordinates based on $k$, a parallel screen frame, and a
parallel complementary null vector.  The weighted Taylor expansion of the
metric has degree-two part
\[
 2\,du\,dv+\delta_{ij}\,dx^i dx^j
 -g(R(e_i,k)k,e_j)x^ix^j\,du^2;
\]
all remaining terms have strictly larger weighted degree.  This is the standard
covariant expansion of the Penrose limit \cite{Blau2006}.  A change of parallel
screen frame conjugates the profile orthogonally, while a change of full Fermi
realization has the gauge behavior proved in \Cref{sec:weighted-gauge}.  Hence
the displayed tensor is the intrinsic weighted symbol.
\end{proof}

\section{Weighted realizations and higher-order gauge}
\label{sec:weighted-gauge}

The Penrose symbol is canonical.  Its realization in a spacetime neighborhood
requires further choices, recorded below.

\begin{definition}\label{def:weighted-realization}
A \emph{weighted realization} of $W_\gamma$ is a germ along the zero section of
a diffeomorphism
\[
 \Phi:(W_\gamma,\gamma)\longrightarrow(M,\gamma)
\]
which is the identity on $\gamma$ and whose normal derivative respects the null
filtration.  It is an embedding germ of neighborhoods, with no isometry
condition imposed on $g_{\PW}$.
\end{definition}

If $F$ is the transition between two weighted realizations, its weighted
principal part is
\begin{equation}\label{eq:principal-part}
 \sigma_w(F)=\lim_{\epsilon\to0}
       \delta_\epsilon^{-1}\circ F\circ\delta_\epsilon,
\end{equation}
whenever the realizations use the same zero section and affine parameter.  In
local weighted coordinates it has the form
\begin{equation}\label{eq:principal-coordinate-form}
 u'=u,
 \qquad x'=A(u)x,
 \qquad v'=a(u)v+Q_u(x),
\end{equation}
where $A(u)$ is invertible, $a(u)\ne0$, and $Q_u$ is quadratic.  Thus the
frequently used normalization $u'=u$ concerns the principal part.  A general
adapted coordinate change may also contain the higher-weight terms
\begin{equation}\label{eq:adapted-transition}
 u'=u+O_w(1),\qquad
 x'=A(u)x+O_w(2),\qquad
 v'=a(u)v+Q_u(x)+O_w(3).
\end{equation}

An \emph{admissible weighted diffeomorphism germ} is a germ fixing $\gamma$
pointwise and having the form \eqref{eq:adapted-transition} in one, and hence
every, weighted adapted chart.  Equivalently, it preserves each ideal of the
weighted filtration introduced below.  Let $\Gad$ be the group of these germs
and let $G_0$ be its group of principal parts.  Define
\begin{equation}\label{eq:high-gauge}
 \Ghigh=\ker(\sigma_w)
       =\{F\in\Gad:\sigma_w(F)=\operatorname{id}\}.
\end{equation}
There is therefore an exact sequence
\begin{equation}\label{eq:gauge-sequence}
 1\longrightarrow\Ghigh\longrightarrow\Gad
 \xrightarrow{\ \sigma_w\ }G_0\longrightarrow1,
\end{equation}
locally after restricting $G_0$ to the principal parts which occur.  The
pointwise homogeneous structure group is finite-dimensional.  Because its
coefficients may vary with $u$, $G_0$ is the corresponding group of smooth
gauge germs.  It records changes of graded frame and polarization, while
$\Ghigh$ records the noncanonical extension of the same graded soldering to an
actual neighborhood germ.

\begin{lemma}[Weighted functoriality]\label{lem:weighted-functoriality}
Let $T$ be a tensor germ with weighted leading part $\grw(T)$, and let
$F\in\Gad$.  Then
\begin{equation}\label{eq:weighted-functoriality}
 \grw(F^*T)=\sigma_w(F)^*\grw(T).
\end{equation}
\end{lemma}

\begin{proof}
Write $F=\sigma_w(F)+R$ and $T=T_0+S$, where every component of $R$ and $S$
has strictly larger weight than the corresponding leading component.  Every
term in $F^*T$ involving $R$ or $S$ has larger weight, leaving
$\sigma_w(F)^*T_0$ on the associated graded.
\end{proof}

\begin{theorem}[Gauge degeneration]\label{thm:gauge-degeneration}
For every $F\in\Gad$ the conjugates
\[
 F_\epsilon=\delta_\epsilon^{-1}F\delta_\epsilon
\]
converge in $C^\infty$ on compact subsets to $\sigma_w(F)$.  If $g_\epsilon=
\epsilon^{-2}\delta_\epsilon^*g$ converges to $g_{\PW}$ in one weighted
realization, then in the realization changed by $F$ it converges to
\[
 \sigma_w(F)^*g_{\PW}.
\]
In particular, an element of $\Ghigh$ leaves the limiting representative
unchanged.
\end{theorem}

\begin{proof}
Taylor-expand each component of $F$ in the variables of weights one and two.
Conjugation multiplies a monomial by $\epsilon$ to the difference between its
input and output weights.  Admissibility excludes negative differences, the
terms of difference zero give \eqref{eq:principal-coordinate-form}, and all
positive differences vanish with all derivatives.  Moreover,
\[
 \epsilon^{-2}\delta_\epsilon^*(F^*g)
 =F_\epsilon^*(\epsilon^{-2}\delta_\epsilon^*g),
\]
which gives the metric assertion, equivalently
\Cref{lem:weighted-functoriality} applied to the metric symbol.
\end{proof}

\subsection{Approximate dilations}

Let $\mathcal F_\gamma^r$ be the ideal of functions whose weighted Taylor
expansion along $\gamma$ contains no term of weight less than $r$.  Thus, in an
adapted chart, $x^i\in\mathcal F_\gamma^1$, $v\in\mathcal F_\gamma^2$, and
$u$ has weight zero.  Gauge degeneration shows that this filtration is
independent of the adapted realization.  Its associated graded algebra is the
fiberwise polynomial algebra on $W_\gamma$, and its Euler derivation is
\begin{equation}\label{eq:euler-derivation}
 \varepsilon(u)=0,
 \qquad \varepsilon(x^i)=x^i,
 \qquad \varepsilon(v)=2v.
\end{equation}

\begin{definition}\label{def:approximate-dilation}
An \emph{approximate dilation} along $\gamma$ is a vector-field germ $X$ which
preserves every $\mathcal F_\gamma^r$, induces the Euler derivation
\eqref{eq:euler-derivation} on the associated graded algebra, and fixes the
chosen affine parameter: $X(u)=0$.
\end{definition}

\begin{proposition}[Coordinate characterization]
\label{prop:approximate-coordinate}
A vector field is an approximate dilation if and only if, in one and hence
every adapted chart, it has the form
\begin{equation}\label{eq:approximate-coordinate}
 X=2v\partial_v+x^i\partial_{x^i}
   +R_v\partial_v+R_x^i\partial_{x^i},
 \qquad R_v=O_w(3),\quad R_x^i=O_w(2).
\end{equation}
Every weighted first jet inducing the Euler derivation admits such a local
extension.
\end{proposition}

\begin{proof}
The condition on the associated graded says that
$X(x^i)-x^i\in\mathcal F_\gamma^2$ and
$X(v)-2v\in\mathcal F_\gamma^3$; together with $X(u)=0$ this is
\eqref{eq:approximate-coordinate}.  The converse is immediate.  For existence,
choose one adapted realization and start with
$2v\partial_v+x^i\partial_{x^i}$; adding any prescribed higher-weight terms
leaves its induced derivation unchanged.
\end{proof}

Approximate dilation fields determine the same gauge class.  A final
higher-order coordinate change may be needed to obtain the same representative.

\begin{proposition}[Approximate dilations]\label{prop:approximate-dilation}
Let $X$ be an approximate dilation, write
$D=x^i\partial_{x^i}+2v\partial_v$, and set $X=D+R$ as in
\eqref{eq:approximate-coordinate}.  Let $\Phi_t$ be the local flow of $X$ and put
\[
 H_t=\delta_{e^{-t}}^{-1}\circ\Phi_{-t}.
\]
On a sufficiently small germ, $H_t$ converges as $t\to+\infty$ to an element
$H_\infty\in\Ghigh$, and
\begin{equation}\label{eq:approximate-limit}
 e^{2t}\Phi_{-t}^*g\longrightarrow H_\infty^*g_{\PW}.
\end{equation}
Thus an approximate dilation canonically determines the $\Ghigh$-class of the
Penrose germ.  Different extensions may select different representatives of
that class.
\end{proposition}

\begin{proof}
After conjugation by $\delta_{e^{-t}}$, the vector field $R$ acquires a strictly
decaying factor in every weighted homogeneous component.  Variation of
constants gives convergence of $H_t$ and of all derivatives; its principal
part is the identity.  Factoring
$\Phi_{-t}=\delta_{e^{-t}}H_t$ in the rescaled pullback and using
\Cref{prop:penrose-symbol} gives \eqref{eq:approximate-limit}.
\end{proof}

\section{Contact geometry and incidence evaluation}
\label{sec:contact-incidence}

We now pass from one null geodesic to a local family.  All assertions are local
on the set of null geodesics, so we assume that the quotient by the null
geodesic flow is a smooth cooriented contact manifold $\cN$ of dimension
$2n+1$.  Global causality and Hausdorff questions lie outside this local
argument; see
\cite{Low1989,Low2006} for the global theory.

\subsection{The null cone and symplectization}

Let
\begin{equation}\label{eq:null-cotangent-cone}
 \Sigma=\{(x,p)\in T^*M\setminus0:g^{-1}_x(p,p)=0\}
\end{equation}
be one cooriented component of the nonzero null cotangent cone.  Write
$\alpha$ for the tautological form on $T^*M$, $\omega=d\alpha$, and
$H(x,p)=\tfrac12g^{-1}_x(p,p)$.  The characteristic line of
$\omega|_\Sigma$ is generated by the Hamiltonian vector field $X_H$; its
integral curves are the affinely parametrized cotangent lifts of null
geodesics.  The fiber Euler field $\mathcal E=p_a\partial_{p_a}$ rescales their
affine normalization.

\begin{theorem}[Null-cone reduction]\label{thm:null-cone-reduction}
Locally, the quotient
\begin{equation}\label{eq:scaled-ray-space}
 \mathcal G=\Sigma/\langle X_H\rangle
\end{equation}
is symplectic, and its quotient by the positive flow of $\mathcal E$ is the
contact manifold $\cN$ of unparametrized null geodesics.  Equivalently,
$\mathcal G$ is the symplectization of $\cN$: after choosing a positive contact
form $\theta$ on $\cN$ there is a local identification
\begin{equation}\label{eq:symplectization}
 (\mathcal G,\omega_{\mathcal G})
 \cong(\mathbb R_{>0}\times\cN,d(r\theta)).
\end{equation}
The coordinate $r$ is the scale of the null covector and therefore the scale
of the affine tangent $k=g^{-1}(p)$.
\end{theorem}

\begin{proof}
The kernel of $\omega|_{T\Sigma}$ is the characteristic line
$\langle X_H\rangle$, so local characteristic reduction gives the symplectic
space $\mathcal G$.  Since $H$ is homogeneous of degree two,
$[\mathcal E,X_H]=X_H$; consequently fiber dilations preserve characteristic
leaves and descend to $\mathcal G$.  The tautological form is homogeneous of
degree one.  Choosing a section of the positive dilation quotient therefore
writes its descendant as $r\theta$ and gives
\eqref{eq:symplectization}.  Quotienting by $r$ leaves the contact hyperplane
$\ker\theta$ on $\cN$.
\end{proof}

The contact quotient retains the ray and discards its affine scale.  A contact
scale restores the normalization, and an incident point $x\in\gamma$ supplies
an origin when one is needed.

\subsection{The incidence bundles}

Set
\begin{equation}\label{eq:incidence}
 \cU=\{(x,[\gamma])\in M\times\cN:x\in\gamma\},
 \qquad
 \rho:\cU\to M,
 \quad \pi:\cU\to\cN.
\end{equation}
The fibers of $\pi$ are the null geodesics and the fibers of $\rho$ are their
skies.  On $\cU$ define the null line, screen bundle, and weight-two line by
\begin{equation}\label{eq:incidence-flag}
 \cK=d\rho(\ker d\pi)\subset\rho^*TM,
 \qquad
 \cE=\cK^\perp/\cK,
 \qquad
 \cL=\rho^*TM/\cK^\perp.
\end{equation}
Thus $\cE[1]\oplus\cL[2]$ is the family version of
\eqref{eq:weighted-normal}.

\subsection{Jacobi fields and skies}

Let $p=[\gamma]\in\cN$.  A tangent vector to $\cN$ is represented by a Jacobi
field arising from a variation through null geodesics, modulo the tangential
Jacobi fields $(a+bu)k$.  For such a representative, $g(J,k)$ is constant along
$\gamma$.  The canonical contact hyperplane is
\begin{equation}\label{eq:contact-hyperplane}
 \cC_p=\{[J]\in T_p\cN:g(J,k)=0\}.
\end{equation}
The Wronskian induces the conformal symplectic Levi form on $\cC_p$.

Projection to the screen gives a canonical identification
\begin{equation}\label{eq:contact-transverse-jacobi}
 \cC_p\xrightarrow{\sim}\mathcal J_\gamma^\perp.
\end{equation}
Indeed, an abreast representative takes values in $k^\perp$ and hence projects
to a solution of \eqref{eq:jacobi}.  The kernel consists of Jacobi fields
$f(u)k$ with $f''=0$, the tangential fields already removed in
$T_p\cN$.  Under \eqref{eq:contact-transverse-jacobi}, the Levi form agrees,
up to the contact scale, with the Wronskian \eqref{eq:wronskian}.  This also
distinguishes the full null-variation Jacobi space used to describe
$T_p\cN$ from the $2n$-dimensional transverse solution space.

For $q=(x,p)\in\cU$, let
\begin{equation}\label{eq:sky-lagrangian}
 \cH_q=d\pi(\ker d\rho)_q
       =T_p\Sky(x)\subset\cC_p.
\end{equation}
In Jacobi language, $\cH_q$ consists of the classes which have a representative
vanishing at $x$.  It has dimension $n$ and is Lagrangian for the Levi form.
As $x$ moves along $\gamma$, $\cH_q$ is the Hamiltonian curve of
vanishing Jacobi data.

Evaluation at the incident point relates contact phase space to spacetime.

\begin{theorem}[Incidence evaluation]\label{thm:incidence-evaluation}
At $q=(x,p)\in\cU$, evaluation of Jacobi fields induces a canonical exact
sequence
\begin{equation}\label{eq:evaluation-exact}
 0\longrightarrow\cH_q\longrightarrow T_p\cN
 \xrightarrow{\ \operatorname{ev}_q\ }T_xM/\cK_q
 \longrightarrow0.
\end{equation}
It respects the filtrations and therefore induces canonical isomorphisms
\begin{equation}\label{eq:evaluation-graded}
 \cC_p/\cH_q\xrightarrow{\sim}\cE_q,
 \qquad
 T_p\cN/\cC_p\xrightarrow{\sim}\cL_q.
\end{equation}
The first is the weight-one soldering and the second is the weight-two
soldering.
\end{theorem}

\begin{proof}
For a Jacobi representative $J$, set
$\operatorname{ev}_q[J]=[J(x)]\bmod\cK_q$.  Replacing $J$ by
$J+(a+bu)k$ leaves the result unchanged.  The kernel consists of
variations which can be represented by a Jacobi field vanishing at $x$, hence
is $\cH_q$.  Arbitrary initial value at $x$ can be completed to Jacobi initial
data satisfying the null-variation constraint, so evaluation is surjective.

If $[J]\in\cC_p$, then $g(J,k)=0$ and its value lies in
$\cK_q^\perp/\cK_q=\cE_q$.  Conversely, this orthogonality is equivalent to
the contact condition.  Taking the subquotient and quotient of
\eqref{eq:evaluation-exact} gives \eqref{eq:evaluation-graded}.
\end{proof}

The theorem is independent of a parameter, screen, contact form, or
polarization.  It is the canonical part of the soldering construction.

\subsection{Contact-scale jets}

Let
\begin{equation}\label{eq:contact-line-scale}
 \mathcal Z=T\cN/\cC,
 \qquad
 \cS=(\mathcal Z^*)_+
\end{equation}
be the contact line and its positive scale bundle.  A local section of $\cS$ is
a positive contact form $\theta$ with $\ker\theta=\cC$.  Its first jet
$\eta=j^1_p\theta$ determines both $\theta_p$ and $d\theta_p$.  Hence it
determines
\begin{equation}\label{eq:reeb}
 \omega_\eta=d\theta_p|_{\cC_p},
 \qquad
 \theta_p(R_\eta)=1,
 \qquad
 \iota_{R_\eta}d\theta_p=0.
\end{equation}
In particular, $\eta$ splits
$T_p\cN=\cC_p\oplus\mathbb R R_\eta$.

Choose Darboux coordinates $(x^i,\xi_i,v)$ centered at $p$ in which
\begin{equation}\label{eq:contact-darboux}
 \theta=dv+\xi_i\,dx^i.
\end{equation}
At the origin the Reeb vector is $\partial_v$, the contact hyperplane has
coordinates $(x,\xi)$, and the osculating dilations are
\begin{equation}\label{eq:heisenberg-dilation}
 \delta_\epsilon^{\mathrm H}(x,\xi,v)
 =(\epsilon x,\epsilon\xi,\epsilon^2v).
\end{equation}
Their grading field is
\begin{equation}\label{eq:heisenberg-euler}
 D_{\mathrm H}=x^i\partial_{x^i}
 +\xi_i\partial_{\xi_i}+2v\partial_v.
\end{equation}

\begin{theorem}[Heisenberg tangent model]\label{thm:heisenberg-model}
The nilpotent approximation of $(\cN,\cC)$ at $p$ is the graded Heisenberg
algebra
\begin{equation}\label{eq:heisenberg-algebra}
 \mathfrak h_p=\cC_p[1]\oplus\mathcal Z_p[2],
 \qquad
 \theta_p([X,Y]_{\mathrm{gr}})=-d\theta_p(X,Y).
\end{equation}
Its dilation is intrinsic and acts with weights one and two.  A contact scale
identifies the central line with $\mathbb R R_\eta$; the grading is independent
of this identification.
\end{theorem}

\begin{proof}
The contact distribution is bracket-generating of step two, and its Levi
bracket is the quotient of $[X,Y]$ in $T_p\cN/\cC_p$.  Since the Levi form is
nondegenerate, this is the Heisenberg algebra.  In
\eqref{eq:contact-darboux} the bracket and the dilations are the standard ones,
and \eqref{eq:heisenberg-euler} represents the grading derivation.
\end{proof}

\begin{theorem}[Realized grading from a contact-scale jet]
\label{thm:grading-from-contact-jet}
A jet $\eta=j^1_p\theta$ canonically determines a weighted first-jet class
$[D_\eta]$ of vector-field germs at $p$ with the following properties:
\begin{enumerate}[label=(\roman*)]
\item its induced derivation on $\mathfrak h_p$ is the intrinsic grading
      derivation;
\item its degree-two line is represented by $R_\eta$;
\item in any Darboux realization inducing the same symplectic frame,
      $[D_\eta]$ is represented by $D_{\mathrm H}$ modulo terms which induce
      zero on the osculating graded algebra.
\end{enumerate}
Consequently contact forms with the same first jet determine the same realized
grading class.
\end{theorem}

\begin{proof}
The first jet determines $\theta_p$ and $d\theta_p$, hence the Levi form and
the Reeb value through \eqref{eq:reeb}.  The graded algebra then has its
canonical Euler derivation, represented by \eqref{eq:heisenberg-euler} in a
Darboux chart.  Any two such representatives differ by a vector field acting
trivially on the associated graded.  All these data depend only on
$j^1_p\theta$.
\end{proof}

\begin{proposition}[First-jet principle]\label{prop:first-jet-principle}
Fix $\eta=j^1_p\theta$ and a symplectic frame of
$(\cC_p,\omega_\eta)$.  Any two Darboux realization germs inducing that frame,
the scale $\theta_p$, and the Reeb representative $R_\eta$ have transition with
identity Heisenberg principal part.  A change of symplectic frame changes the
principal part by its homogeneous contact lift.  After an incident sky is fixed,
the allowed frame changes reduce to the sky parabolic
\eqref{eq:parabolic} below.
\end{proposition}

\begin{proof}
In Darboux coordinates $(x,\xi,v)$, with weights $(1,1,2)$ and contact form
$dv+\xi^Tdx$, the weight-one part of a transition is its derivative on
$\cC_p$, while the central linear part is fixed by the scale and Reeb vector.
A homogeneous quadratic change of $v$ forces, by the contact equation, the
corresponding linear shear of $(x,\xi)$; hence it is already detected by the
symplectic frame.  If the frame is fixed, no nontrivial homogeneous term
remains and the principal part is the identity.  If the frame changes, the
same calculation gives its homogeneous contact lift.  Requiring the reference
$V^*$-Lagrangian to map to the sky gives its stabilizer.
\end{proof}

The value $\theta_p$ trivializes the contact line.  Through the second map in
\eqref{eq:evaluation-graded} and the metric pairing
$\cK_q\otimes\cL_q\to\mathbb R$, it also fixes the scale of the affine tangent
to $\gamma$.  This normalization is independent of the incident point.  To see
this, choose any affine parameter $u$ and put $k_0=d/du$.  If $J$ is a
null-variation Jacobi field, then
\[
 \frac{d}{du}g(J,k_0)=g(\nabla_{k_0}J,k_0)=0.
\]
The evaluation isomorphisms therefore identify $\theta_p$ with one constant
rescaling of $k_0$ along the whole ray.  The first derivative of the scale fixes
the Reeb representative of the degree-two class.  Thus
\Cref{thm:null-cone-reduction} explains the affine scale, while
\Cref{thm:heisenberg-model,thm:grading-from-contact-jet} explain the weighted
realization selected by its first jet.

Extending the evaluation maps \eqref{eq:evaluation-graded} to full
neighborhoods requires higher jets.  For a fixed contact-scale first jet, all
such extensions differ by $\Ghigh$ and define one gauge class.

\begin{corollary}[Phase-space Penrose bridge]
\label{cor:phase-space-bridge}
Pull the Heisenberg tangent model to
$\widehat{\cU}=\cU\times_{\cN}J^1\cS$.  At
$(q,\eta)$ its sky quotient is
\begin{equation}\label{eq:heisenberg-sky-quotient}
 (\cC_p/\cH_q)[1]\oplus\mathcal Z_p[2]
 \xrightarrow{\sim}\cE_q[1]\oplus\cL_q[2].
\end{equation}
The Heisenberg dilation descends to the Penrose dilation on the right.  After
adjoining the weight-zero line $\cK_q$, the curvature symbol on this graded
configuration space is the plane-wave metric $g_{\PW}$ of
\eqref{eq:brinkmann-limit}.
\end{corollary}

\begin{proof}
The two components of \eqref{eq:heisenberg-sky-quotient} are the isomorphisms
\eqref{eq:evaluation-graded}.  Their weights are inherited from
the Heisenberg grading.  The null-cone reduction supplies the affine scale,
and \Cref{prop:penrose-symbol} identifies the resulting metric symbol.
\end{proof}

\begin{proposition}[Spacetime extensions of the realized grading]
\label{prop:j1-approximate-extension}
For every $(q,\eta)\in\widehat{\cU}$, the sky quotient of $[D_\eta]$ admits a
local extension to an approximate dilation along the corresponding affinely
normalized null geodesic.  Any two such extensions induce the same Euler
derivation and their renormalized flows differ, in the limit, by an element of
$\Ghigh$.
\end{proposition}

\begin{proof}
Choose a sky-adapted Darboux frame and a weighted tubular realization.  The
quotient of \eqref{eq:heisenberg-euler} is
$x^i\partial_{x^i}+2v\partial_v$; extending it constantly in the affine
$u$-direction gives an approximate dilation.  A second extension has the same
weighted first jet and hence differs by the remainder allowed in
\eqref{eq:approximate-coordinate}.  Apply
\Cref{prop:approximate-dilation} to their flows.
\end{proof}

\section{The bundle of soldered Penrose germs}
\label{sec:soldered-bundle}

Put
\begin{equation}\label{eq:bases}
 \mathcal B=J^1\cS,
 \qquad
 \widehat{\cU}=\cU\times_{\cN}\mathcal B.
\end{equation}
All the bundles in \eqref{eq:incidence-flag} will be pulled back to
$\widehat{\cU}$ without a change of notation.  For fixed
$\eta\in\mathcal B$ over $p=[\gamma]$, the value of the scale fixes an affine
normalization of $\gamma$ and therefore fixes the curvature profile
\eqref{eq:profile}.  Varying $(p,\eta)$ gives a smooth bundle
\begin{equation}\label{eq:pw-bundle}
 \PW(M,g)\longrightarrow\mathcal B
\end{equation}
whose fiber is the plane-wave metric germ on
$\cE[1]\oplus\cL[2]$ along the corresponding null geodesic.  The dependence on
the value of $\eta$ fixes the curvature symbol; its derivative part enters the
realization of the germ.

\subsection{The unpolarized Penrose bundle}

Let $V=\mathbb R^n$ and give $V\oplus V^*$ its standard symplectic form.  Over
$\eta=j^1_p\theta\in\mathcal B$, define $\mathsf P_0$ to be the bundle of
symplectic frames
\begin{equation}\label{eq:symplectic-frame}
 s:V\oplus V^*\xrightarrow{\sim}(\cC_p,\omega_\eta).
\end{equation}
The contact scale identifies the center of the standard Heisenberg algebra
\begin{equation}\label{eq:standard-heisenberg}
 \mathbb H^{\mathrm{std}}=(V\oplus V^*)[1]\oplus\mathbb R_z[2]
\end{equation}
with $\mathcal Z_p[2]$ and represents it by $R_\eta$.  Thus a point of
$\mathsf P_0$ is an unpolarized Heisenberg/Penrose frame, including its realized
central direction.

\begin{proposition}\label{prop:unpolarized-torsor}
The bundle $\mathsf P_0\to\mathcal B$ is a principal
$\Sp(2n,\mathbb R)$-bundle.  Its associated weighted Heisenberg bundle
\begin{equation}\label{eq:unpolarized-model-bundle}
 \mathsf H_0=\mathsf P_0\times_{\Sp(2n,\mathbb R)}
 \mathbb H^{\mathrm{std}}
\end{equation}
is intrinsic over $J^1\cS$ and is independent of a Lagrangian polarization.
\end{proposition}

\begin{proof}
Once $\eta$ fixes $\omega_\eta$ and the central scale, two frames differ
uniquely by a symplectic automorphism of $V\oplus V^*$.  The associated-bundle
statement follows.  The first-jet independence is
\Cref{prop:first-jet-principle}.
\end{proof}

The bundle $\mathsf H_0$ is unpolarized.  Incidence reduces it by the sky
Lagrangian; a Rosen polarization is an additional choice.

\subsection{The sky-parabolic bundle}

We next describe the finite-dimensional structure group of the principal gauge
over incidence.  After pullback to $\widehat{\cU}$, the sky Lagrangian defines
the reduction
\begin{equation}\label{eq:sky-reduction}
 \mathsf P_{\mathrm{sky}}
 =\{s\in\mathsf P_0:s(V^*)=\cH_q\}.
\end{equation}
Its structure group is the stabilizer of the \emph{$V^*$-Lagrangian}, namely
\begin{equation}\label{eq:parabolic}
 P=\left\{
 \begin{pmatrix}A&0\\ C&A^{-T}\end{pmatrix}:
 A\in\operatorname{GL}(V),\ A^TC=(A^TC)^T
 \right\}\subset\Sp(V\oplus V^*).
\end{equation}
Writing $B=-A^TC$, so that $B=B^T$, its contact lift in Darboux coordinates is
\begin{equation}\label{eq:contact-parabolic-action}
 x'=Ax,
 \qquad
 \xi'=A^{-T}(\xi-Bx),
 \qquad
 v'=v+\tfrac12x^TBx.
\end{equation}
Indeed $dv'+(\xi')^Tdx'=dv+\xi^Tdx$.  Formula
\eqref{eq:contact-parabolic-action} also fixes the sign of the quadratic term,
which is otherwise easy to obscure in an abstract parabolic notation.

Taking the sky-quotient configuration variables $(x,v)$ gives the associated
bundle of weighted homogeneous model germs
\begin{equation}\label{eq:configuration-bundle}
 \mathsf Q=\mathsf P_{\mathrm{sky}}
 \times_P\bigl(V[1]\oplus\mathbb R[2]\bigr),
\end{equation}
where $P$ acts by the $(x,v)$ part of
\eqref{eq:contact-parabolic-action}.  Because the $v$-transition includes a
quadratic term, $\mathsf Q$ is naturally a weighted bundle.  The
$P$-equivalence class of its homogeneous charts
is the full weighted principal soldering datum; denote it by
$\operatorname{Sol}_0$.  This datum includes the quadratic $v$-shift in
\eqref{eq:contact-parabolic-action}.  Full neighborhood maps require
higher-order data.  The linear associated graded is ordinary, and incidence
evaluation gives
\begin{equation}\label{eq:graded-soldering}
 \operatorname{Sol}_{\mathrm{gr}}:
 \grw(\mathsf Q)\xrightarrow{\sim}\cE[1]\oplus\cL[2]=:\cW.
\end{equation}
The linear map $\operatorname{Sol}_{\mathrm{gr}}$ omits the quadratic term
$Q_u(x)$ in \eqref{eq:principal-coordinate-form}.  Changes of this term belong
to the parabolic principal gauge.  Although a general element of $G_0$ may
contain $v\mapsto a(u)v$, $\operatorname{Sol}_0$ includes the scaled
identification of the central line with $\cL[2]$.  Preserving this identification
forces $a(u)=1$.  The remaining homogeneous frame changes are the contact lifts
of $P$ and are already absorbed in the associated bundle $\mathsf Q$.  Thus the
relative principal part of two neighborhood extensions of the same intrinsic
$\operatorname{Sol}_0$ is the identity.  The $\Ghigh$-torsor is defined only
after this complete datum has been fixed.

\subsection{Local transitions and the principal cocycle}

Choose local sections $s_\alpha$ of $\mathsf P_{\mathrm{sky}}$.  On an overlap
they differ by a unique smooth map $g_{\alpha\beta}$ with values in $P$,
\begin{equation}\label{eq:parabolic-cocycle}
 g_{\alpha\beta}
 =\begin{pmatrix}A_{\alpha\beta}&0\\
 C_{\alpha\beta}&A_{\alpha\beta}^{-T}\end{pmatrix},
 \qquad
 B_{\alpha\beta}=-A_{\alpha\beta}^TC_{\alpha\beta}
 =B_{\alpha\beta}^T,
\end{equation}
and $g_{\alpha\beta}g_{\beta\gamma}=g_{\alpha\gamma}$ on triple
overlaps.  The corresponding fiber coordinates obey
\begin{equation}\label{eq:local-principal-transitions}
 x_\beta=A_{\alpha\beta}x_\alpha,
 \qquad
 \xi_\beta=A_{\alpha\beta}^{-T}
 (\xi_\alpha-B_{\alpha\beta}x_\alpha),
 \qquad
 v_\beta=v_\alpha+\tfrac12x_\alpha^TB_{\alpha\beta}x_\alpha.
\end{equation}
These are fiberwise formulas over $\widehat{\cU}$; their coefficients may vary
smoothly with the incident point, including along the affine $u$-direction.
They commute with the weights $(1,1,2)$ and therefore constitute the complete
homogeneous principal cocycle.  Passing to $(x,v)$ gives the transition
functions of $\mathsf Q$.

If local weighted tubular realizations are appended to these homogeneous
charts, their transitions have the fuller form
\begin{equation}\label{eq:full-local-transitions}
 u'=u+O_w(1),\qquad
 x'=A(u)x+O_w(2),\qquad
 v'=v+\tfrac12x^TB(u)x+O_w(3).
\end{equation}
The first line of data is the parabolic cocycle
\eqref{eq:local-principal-transitions}; the remainders belong to the
higher-order realization problem.  This is the local content of the exact
sequence \eqref{eq:gauge-sequence}.

\subsection{Rosen reductions}
\label{sec:rosen-reductions}

The sky Lagrangian $\cH_q$ is the kernel of evaluation.  A Rosen gauge requires
a complementary Lagrangian $\Lambda\subset\cC_p$ satisfying
\begin{equation}\label{eq:two-lagrangians}
 \cC_p=\Lambda\oplus\cH_q.
\end{equation}
Introduce the optional polarized base
\begin{equation}\label{eq:rosen-base}
 \mathcal B_R=J^1\cS\times_{\cN}\operatorname{LG}(\cC)
\end{equation}
and pull it to incidence.  On the open locus where
\eqref{eq:two-lagrangians} holds, frames satisfying
\begin{equation}\label{eq:double-lagrangian-frame}
 s(V^*)=\cH_q,
 \qquad s(V)=\Lambda
\end{equation}
form a principal $\operatorname{GL}(V)$-bundle.  A Rosen polarization is
therefore a reduction of the sky-parabolic bundle to the Levi subgroup
$\operatorname{GL}(V)\subset P$, whose inclusion is
\begin{equation}\label{eq:rosen-reduction}
 \operatorname{GL}(V)\hookrightarrow P.
\end{equation}

The quotient $P/\operatorname{GL}(V)$ is the affine space of Lagrangian
complements to $V^*$, modeled on $\operatorname{Sym}^2(V^*)$.  In
\eqref{eq:local-principal-transitions}, the symmetric matrix $B$ records the
change of complement and the associated quadratic shift of $v$.  The three
levels are:
\begin{enumerate}[label=(\roman*)]
\item $\mathsf P_0$ is the intrinsic unpolarized symplectic bundle;
\item $\mathsf P_{\mathrm{sky}}$ is the canonical parabolic reduction over
      incidence;
\item a transverse $\Lambda$ reduces it further and supplies a Rosen
      realization until a caustic is reached.
\end{enumerate}

\subsection{Associated models over incidence}

Let $\mathscr K=\ker d\pi\subset T\cU$ be the line tangent to the geodesic
fibers.  The differential $d\rho$ identifies it with the tautological ambient
null line $\cK$.  Pulling \eqref{eq:unpolarized-model-bundle} to
$\widehat{\cU}$ gives the phase-space family
\begin{equation}\label{eq:pulled-heisenberg-model}
 \widehat{\mathsf H}_0
 =\widehat{\cU}\times_{\mathcal B}\mathsf H_0.
\end{equation}
The sky quotient of its weight-one layer, together with its central line, is
the configuration bundle $\mathsf Q$.  Incidence evaluation therefore extends
\eqref{eq:graded-soldering} to the full graded tangent along the ray:
\begin{equation}\label{eq:full-graded-soldering}
 \widetilde{\operatorname{Sol}}_{\mathrm{gr}}:
 \mathscr K[0]\oplus\grw(\mathsf Q)
 \xrightarrow{\sim}
 \cK[0]\oplus\cE[1]\oplus\cL[2].
\end{equation}
The affine coordinate $u$ belongs to $\mathscr K$, hence to the incidence
family.  Together with the $n+1$ weighted configuration directions it accounts
for the full spacetime dimension.

\begin{proposition}[Coordinate form of the soldering]
\label{prop:coordinate-soldering}
Let $\Phi$ be a local weighted realization of the complete principal datum
$\operatorname{Sol}_0$.  Then
\begin{equation}\label{eq:coordinate-symbol-remainder}
 \Phi^*g=g_{\PW}+h,
 \qquad
 \epsilon^{-2}\delta_\epsilon^*h\longrightarrow0
 \quad\text{in }C^\infty_{\mathrm{loc}}.
\end{equation}
If $\widetilde\Phi$ is an arbitrary second adapted realization and
$F=\Phi^{-1}\widetilde\Phi$, then its limiting representative is
$\sigma_w(F)^*g_{\PW}$.  If $\widetilde\Phi$ realizes the same
$\operatorname{Sol}_0$, then $\sigma_w(F)=\operatorname{id}$ and the two
limits agree.
\end{proposition}

\begin{proof}
The first assertion is the definition of the weighted metric symbol together
with \Cref{prop:penrose-symbol}.  For the second, apply
\Cref{thm:gauge-degeneration,lem:weighted-functoriality}.  Agreement on the
scaled weight-two identification forces $a(u)=1$.  Agreement on the remaining
complete principal datum, including its quadratic $v$-term, then gives
$\sigma_w(F)=\operatorname{id}$, which is equivalent to $F\in\Ghigh$.
\end{proof}

\subsection{The canonical soldering theorem}

Let $\operatorname{Real}(\mathsf Q,M)$ denote the bundle of germs of weighted
realizations whose complete weighted principal datum is
$\operatorname{Sol}_0$.  The group $\Ghigh$ acts fiberwise on these
embedding-germ representatives.

\begin{theorem}[Canonical soldering theorem]\label{thm:canonical-soldering}
Let $(M,g)$ be a Lorentzian manifold for which a cooriented local null-geodesic
space $\cN$ is smooth.  Then the following objects are canonical and natural
under local isometries.
\begin{enumerate}[label=(\alph*)]
\item The Penrose plane-wave germs form the smooth bundle
      $\PW(M,g)\to J^1\cS$ of \eqref{eq:pw-bundle}.
\item The realized Heisenberg frames form the intrinsic unpolarized principal
      bundle $\mathsf P_0\to J^1\cS$, with associated weighted model
      $\mathsf H_0$.  No polarization is required for these objects.
\item Over $\widehat{\cU}=\cU\times_{\cN}J^1\cS$, the sky-parabolic bundle
      defines the weighted principal datum $\operatorname{Sol}_0$, including
      its homogeneous quadratic $v$-gauge.  Incidence evaluation identifies
      its linear associated graded by the canonical soldering
      $\operatorname{Sol}_{\mathrm{gr}}$ in
      \eqref{eq:graded-soldering}; adjoining the incidence-fiber line gives
      $\widetilde{\operatorname{Sol}}_{\mathrm{gr}}$ in
      \eqref{eq:full-graded-soldering}.
\item The full realizations extending the same complete principal datum
      $\operatorname{Sol}_0$ form a
      $\Ghigh$-torsor
      \begin{equation}\label{eq:realization-torsor}
       \operatorname{Real}(\mathsf Q,M)
       \longrightarrow\widehat{\cU}.
      \end{equation}
      The canonical datum is their orbit $[\operatorname{Sol}]$; no member of
      this orbit is distinguished.
\item For every representative $\Phi$ of $[\operatorname{Sol}]$, the ambient
      metric has Penrose symbol $g_{\PW}$:
      \begin{equation}\label{eq:main-limit}
       \lim_{\epsilon\to0}
       \epsilon^{-2}\delta_\epsilon^*(\Phi^*g)=g_{\PW}.
      \end{equation}
      Every representative gives this limit.  The flow of an approximate
      dilation may instead give $H^*g_{\PW}$ for $H\in\Ghigh$, in the same
      gauge class.
\end{enumerate}
Thus the bundle of Penrose-limit germs is canonically soldered into spacetime
as a gauge class of weighted embedding germs.  The canonical map to the
ambient null flag is the associated-graded isomorphism in (c), and the
neighborhood-level maps form the torsor in (d).  No canonical isometric
embedding of a Penrose plane wave into $M$ is asserted.
\end{theorem}

\begin{proof}
The null line, screen quotient, curvature endomorphism, and induced screen
connection vary smoothly on $\cU$.  A contact-scale value normalizes the affine
tangent, so \Cref{prop:penrose-symbol} produces the smooth family in (a).

For (b), \Cref{thm:grading-from-contact-jet,prop:first-jet-principle} show that
the Heisenberg frame construction depends only on $j^1\theta$; the ordinary
symplectic frame argument gives the principal $\Sp(2n,\mathbb R)$-action.

For (c), the sky-parabolic construction and its contact lift define
$\operatorname{Sol}_0$ independently of a symplectic frame: a frame change
acts on both descriptions by \eqref{eq:contact-parabolic-action}.  Its
weight-one quotient is $\cC_p/\cH_q$ and its central weight-two line is
$T_p\cN/\cC_p$.  The two canonical isomorphisms of
\Cref{thm:incidence-evaluation} give the linearized map
$\operatorname{Sol}_{\mathrm{gr}}$.

Weighted tubular realizations with prescribed weighted principal data exist
locally.  If $\Phi_1$ and $\Phi_2$ are two which extend the same
$\operatorname{Sol}_0$, then $F=\Phi_1^{-1}\Phi_2$ has identity weighted
principal part;
hence $F\in\Ghigh$.  Conversely, composition with any $F\in\Ghigh$ preserves
the principal part.  This action is free and transitive on germs, proving (d).
Finally (e) follows from \Cref{prop:penrose-symbol,thm:gauge-degeneration}; the
last clause follows from \Cref{prop:approximate-dilation}.
\end{proof}

\subsection{Classical representatives and what they choose}

Null-Fermi and Rosen coordinates select representatives of the bundle above.
A null-Fermi construction chooses an affine
parameter, a parallel orthonormal screen frame, a complementary null vector,
and a tubular extension.  It thereby produces a Brinkmann representative of
the weighted realization class, in which the profile is the curvature matrix
\eqref{eq:profile}.  Changing the screen frame acts through the surviving
homogeneous frame gauge; changing the tubular extension while holding its
complete principal part fixed acts through $\Ghigh$.  Thus the covariant
null-Fermi expansion of \cite{Blau2006} computes a representative of the
symbol and its higher weighted corrections, whereas the soldering theorem
identifies the orbit to which that representative belongs.

A Rosen construction makes an additional choice.  It selects a
Lagrangian $\Lambda$ transverse to the moving sky and integrates the
corresponding Jacobi fields to an adapted null congruence.  The Lagrangian
choice is the Levi reduction $\operatorname{GL}(V)\hookrightarrow P$ of
\eqref{eq:rosen-reduction}; changing that complement produces the symmetric
shear and quadratic $v$-shift in
\eqref{eq:contact-parabolic-action}.  Extending the same homogeneous Rosen
data away from the ray again introduces only $\Ghigh$-freedom.  Hence
Brinkmann and Rosen charts represent the same weighted plane-wave germ.  At a
caustic the Rosen chart fails while the unpolarized bundle remains smooth.

The resulting hierarchy is
\begin{equation}\label{eq:representative-hierarchy}
 \text{weighted realization germ}
 \longmapsto \operatorname{Sol}_0
 \longmapsto \operatorname{Sol}_{\mathrm{gr}}.
\end{equation}
The first arrow forgets the $\Ghigh$-extension, while the second forgets the
homogeneous quadratic $v$-term.  The theorem canonically determines the
$\Ghigh$-orbit over the middle object and the final graded isomorphism; a
coordinate gauge is a choice made above them.

\section{Rosen gauges, caustics, and Legendrian families}
\label{sec:rosen-legendrian}

The unpolarized theorem is global on every local smooth patch of $\cN$.  A
Rosen presentation is more local because it requires a Lagrangian complement
to the moving sky.  Such presentations are also the natural form for several
field-propagation and Hamilton--Jacobi constructions in the preceding work
\cite{HollandSparlingIV,HollandSparlingV}.

\subsection{Rosen transversality}

Fix $p=[\gamma]$ and a Lagrangian subspace
$\Lambda\subset\cC_p$.  Along an interval of $\gamma$ define
\begin{equation}\label{eq:rosen-transverse-locus}
 I_\Lambda=\{x\in\gamma:\Lambda\cap\cH_{(x,p)}=0\}.
\end{equation}

\begin{proposition}[Rosen locus]\label{prop:rosen-locus}
On $I_\Lambda$, incidence evaluation identifies
$\Lambda\cong\cC_p/\cH_{(x,p)}\cong\cE_{(x,p)}$.  A basis
$J_1,\ldots,J_n$ of $\Lambda$, regarded as transverse Jacobi fields, gives the
Rosen metric
\begin{equation}\label{eq:rosen-from-lagrangian}
 g_R=2\,du\,dv+h_{ij}(u)\,dy^i dy^j,
 \qquad h_{ij}(u)=g(J_i(u),J_j(u)).
\end{equation}
This coordinate realization degenerates exactly when
$\Lambda\cap\cH_{(x,p)}\ne0$.
\end{proposition}

\begin{proof}
The first assertion is \Cref{thm:incidence-evaluation} restricted to a
complement of its kernel.  The Wronskian vanishes on $\Lambda$, which is the
integrability condition for the Rosen cross terms to vanish.  The evaluation
matrix loses rank precisely when a nonzero member of $\Lambda$ vanishes at
$x$, equivalently when the stated intersection is nonzero.
\end{proof}

The soldered Penrose germ extends across this locus; the chosen parabolic
reduction has ceased to be transverse to the sky curve.

\subsection{What a Legendrian family supplies}

Let $L\subset\cN$ be a Legendrian submanifold of dimension $n$.

\begin{proposition}[Legendrian polarization]
\label{prop:legendrian-polarization}
The tangent bundle $TL\subset\cC|_L$ is a smooth Lagrangian subbundle.  Hence
$L$ determines a canonical section
\begin{equation}\label{eq:legendrian-lg-section}
 \lambda_L:L\longrightarrow\operatorname{LG}(\cC)|_L,
 \qquad \lambda_L(p)=T_pL.
\end{equation}
\end{proposition}

\begin{proof}
The Legendrian condition gives $TL\subset\cC|_L$.  The Levi form vanishes on
$TL$, and $TL$ has half the rank of $\cC$, so it is maximally isotropic and
hence Lagrangian.
\end{proof}

A section $\eta:L\to J^1\cS|_L$ supplies the remaining contact-scale datum.
Together they define the polarized Legendrian Penrose datum
\begin{equation}\label{eq:legendrian-penrose-datum}
 b_L=(\eta,\lambda_L):L\longrightarrow\mathcal B_R|_L.
\end{equation}
This is a coherent family of Rosen polarizations; at an incident point it gives
a Rosen chart only where $T_pL\cap\cH_q=0$.

Let
\begin{equation}\label{eq:legendrian-incidence}
 \cU_L=\pi^{-1}(L).
\end{equation}
This has dimension $n+1$.

\begin{proposition}[Legendrian incidence]\label{prop:legendrian-incidence}
Away from the locus where
$T_pL\cap\cH_{(x,p)}\ne0$, the map
$\rho:\cU_L\to M$ is an immersion onto a null hypersurface.  It cannot by itself
give a diffeomorphism onto an open spacetime neighborhood.
\end{proposition}

\begin{proof}
The tangent along the $\pi$-fiber maps to the null generator $\cK$.  The
remaining tangent directions are evaluated Jacobi fields from $T_pL$; their
evaluation kernel is $T_pL\cap\cH_{(x,p)}$.  At a regular
point they give $n$ independent screen directions, so the image has dimension
$n+1$ and null normal $\cK$.
\end{proof}

The missing direction is supplied by the contact line.  Choose a local contact
form $\theta$ representing a jet $\eta\in J^1\cS$, let
$\varphi_v^{R_\theta}$ be its Reeb flow, and put
\begin{equation}\label{eq:reeb-thickening}
 L_v=\varphi_v^{R_\theta}(L).
\end{equation}
The incidence of the one-parameter family $\{L_v\}$ has dimension $n+2$.

\begin{proposition}[Reeb thickening]\label{prop:reeb-thickening}
At a point where $T_pL$ is transverse to the sky, the incidence map of
$\{L_v\}$ is a local diffeomorphism onto an open subset of $M$.  It gives a
representative of the soldering class in
\Cref{thm:canonical-soldering}.  Changing the extension of the same first
contact-scale jet changes this representative by higher-order gauge.
\end{proposition}

\begin{proof}
By \Cref{prop:legendrian-incidence}, the $u$-direction and the $n$ Legendrian
directions have rank $n+1$.  The class of the Reeb direction in
$T_p\cN/\cC_p$ evaluates, by \Cref{thm:incidence-evaluation}, to the nonzero
weight-two quotient $\cL_q$.  The tangent space of the null hypersurface is
$\cK_q^\perp$, so this evaluation is automatically transverse and adds the
final rank.  The inverse function theorem gives the local realization.
Two contact forms with the same first jet have the same weighted principal
data, so their transition lies in $\Ghigh$ by
\eqref{eq:high-gauge}.
\end{proof}

The value of the scale normalizes the affine null direction, its first
derivative produces the infinitesimal Reeb thickening, and higher jets vary the
representative of the soldering class.  The construction is therefore naturally
based on $J^1\cS$.

\section{Conformal naturality and projective refinement}
\label{sec:conformal-naturality}

The incidence part of the construction is conformal, while the metric
representatives retain scale dependence.  If $\widehat g=\Omega^2g$, the
unparametrized null geodesics and the incidence correspondence are unchanged.
The contact distribution $\cC\subset T\cN$,
the skies and their Lagrangian tangent spaces $\cH_q$, and the evaluation
filtrations in \Cref{thm:incidence-evaluation} are therefore conformally
natural.  In particular, the linear soldering
$\operatorname{Sol}_{\mathrm{gr}}$ is attached to the conformal null geometry.

\subsection{Scale-dependent metric data}

If $k$ is affine for $g$, then,
up to a constant factor, $\widehat k=\Omega^{-2}k$ is affine for
$\widehat g$.  The screen metric, the representative symplectic form on
$\cC$, and the curvature profile must be transformed accordingly.  A contact
scale records this normalization.  If it is forgotten, the
structure group of the weight-one frame bundle enlarges from
$\Sp(2n,\mathbb R)$ to the conformal symplectic group
$\operatorname{CSp}(2n,\mathbb R)$.

Under conformal rescaling, the metric representative changes together with the
contact scale and affine normalization.  After the scale is forgotten, the
longitudinal parameter geometry is projective.  Its trace normalization comes
from the sky curve.

\subsection{The Schwarzian of the sky curve}
\label{sec:sky-schwarzian}

Fix $p=[\gamma]\in\cN$.  Incidence defines the unparametrized curve
\begin{equation}\label{eq:sky-curve}
 H_\gamma:\gamma\longrightarrow\operatorname{LG}(\cC_p),
 \qquad H_\gamma(x)=\cH_{(x,p)}.
\end{equation}
Under the identification \eqref{eq:contact-transverse-jacobi}, this is the
curve of transverse Jacobi fields vanishing at $x$.  It is a smooth positive
regular curve in the Lagrangian Grassmannian; positivity is the screen metric
written as its tangent form \cite{HollandSparlingIII}.  At a Rosen caustic a
fixed complement ceases to give a graph chart for \eqref{eq:sky-curve}, while
the sky curve remains smooth.

Choose a local parameter $u$ on $\gamma$.  Write
$\mathfrak S_u(H_\gamma)$ for the Grassmannian Schwarzian of the resulting
parameterized curve, as defined in \cite{HollandSparlingIV}.  It is an
endomorphism of the moving Lagrangian $H_\gamma(u)$.  If $u=\phi(s)$, its chain
rule is
\begin{equation}\label{eq:grassmannian-schwarzian-chain}
 \mathfrak S_s(H_\gamma\circ\phi)
 =\bigl(\phi'(s)\bigr)^2
   \mathfrak S_u(H_\gamma)\bigl(\phi(s)\bigr)
  +\{\phi,s\}I,
\end{equation}
where
\begin{equation}\label{eq:ordinary-schwarzian}
 \{\phi,s\}
 =\frac{\phi'''(s)}{\phi'(s)}
  -\frac32\left(\frac{\phi''(s)}{\phi'(s)}\right)^2
\end{equation}
is the ordinary Schwarzian derivative.  The inhomogeneous term in
\eqref{eq:grassmannian-schwarzian-chain} is a scalar endomorphism.

Define
\begin{equation}\label{eq:projective-connection}
 q_u=\frac1n\operatorname{tr}\mathfrak S_u(H_\gamma),
 \qquad
 \mathring{\mathfrak S}_u
 =\mathfrak S_u(H_\gamma)-q_u I.
\end{equation}
Taking the trace and trace-free part of
\eqref{eq:grassmannian-schwarzian-chain} gives
\begin{align}
 q_s
 &=\bigl(\phi'\bigr)^2q_u\circ\phi+\{\phi,s\},
 \label{eq:projective-connection-law}\\
 \mathring{\mathfrak S}_s
 &=\bigl(\phi'\bigr)^2
   \mathring{\mathfrak S}_u\circ\phi.
 \label{eq:tracefree-schwarzian-law}
\end{align}
Thus $q$ transforms as a projective connection, while
$\mathring{\mathfrak S}$ is a trace-free endomorphism-valued quadratic
differential.

For comparison with the ambient curvature, suppose that $u$ is affine for
$g$, put $k=d/du$, and identify $H_\gamma(u)$ with $E_u$ by derivative
evaluation, $J\mapsto\nabla_kJ(u)$.  The curvature calculation of
\cite{HollandSparlingIV},
applied to the transverse Jacobi equation \eqref{eq:jacobi}, gives
\begin{equation}\label{eq:sky-schwarzian-curvature}
 \mathfrak S_u(H_\gamma)=2\mathcal R_\gamma,
 \qquad
 q_u=\frac{2}{n}\operatorname{Ric}_g(k,k).
\end{equation}
If $\mathsf P^g$ denotes the Schouten tensor of $g$, then
$q_u=2\mathsf P^g(k,k)$, since $\dim M=n+2$ and $k$ is null.  This is the
classical conformal projective connection along a null geodesic, expressed
intrinsically through its sky curve.
Equivalently,
\begin{equation}\label{eq:weyl-schwarzian}
 \mathring{\mathfrak S}_u
 =2\left(\mathcal R_\gamma
 -\frac1n\operatorname{Ric}_g(k,k)I_E\right),
\end{equation}
which is twice the Weyl tidal endomorphism.  In the Brinkmann convention of
\eqref{eq:profile}, this is the negative of twice the trace-free part of the
profile $A(u)$.

\begin{theorem}[Projective refinement of a null geodesic]
\label{thm:projective-refinement}
Every unparametrized null geodesic $p=[\gamma]\in\cN$ carries a canonical local
projective structure, natural under local conformal diffeomorphisms.  Its
projective parameters are the local parameters $s$ for which
\begin{equation}\label{eq:projective-parameter}
 \operatorname{tr}\mathfrak S_s(H_\gamma)=0.
\end{equation}
Any two such parameters are related by a fractional-linear transformation.
In a projective parameter the Grassmannian Schwarzian is trace-free and
represents the Weyl tidal profile of the Penrose germ.  These projective
structures vary smoothly with $p$ on every smooth local ray space.
\end{theorem}

\begin{proof}
Equation \eqref{eq:projective-connection-law} makes the collection of local
functions $q_u$ a projective connection on the one-manifold $\gamma$.
Projective parameters exist locally: beginning with any $u$, solve the
third-order Schwarzian equation obtained by setting the right-hand side of
\eqref{eq:projective-connection-law} equal to zero.  If $s$ and $t$ are two
such parameters, applying the same equation to their transition map gives
$\{s,t\}=0$.  The local solutions of the zero-Schwarzian equation are exactly
\begin{equation}\label{eq:fractional-linear-parameters}
 s\longmapsto\frac{as+b}{cs+d},
 \qquad ad-bc\ne0.
\end{equation}
The normalized parameters therefore form a projective atlas.

The unparameterized curve \eqref{eq:sky-curve} is constructed solely from the
conformal contact distribution and the incidence skies.  Its projective
connection, being defined intrinsically by the trace of its Schwarzian, is
therefore natural under conformal diffeomorphisms.  Finally,
\eqref{eq:tracefree-schwarzian-law} and \eqref{eq:weyl-schwarzian} identify the
remaining quadratic differential with the Weyl tidal endomorphism.  Smooth
dependence follows from the smooth sky family and the local Schwarzian formula.
\end{proof}

Without an orientation the transition group in
\eqref{eq:fractional-linear-parameters} is
$\operatorname{PGL}(2,\mathbb R)$.  On the oriented component it is
$\operatorname{PGL}^+(2,\mathbb R)\cong\operatorname{PSL}(2,\mathbb R)$.

\begin{remark}[Agreement with plane-wave conformal transformations]
\label{rem:projective-conformal-check}
The preceding intrinsic construction reproduces the coordinate law proved in
\cite{HollandSparlingII}.  In the sign convention used there, the potential is
\[
 p=\mathcal R_\gamma=-A,
\]
where $A$ is the Brinkmann profile in \eqref{eq:profile}.  If $dU=L^2du$, it
transforms by
\begin{equation}\label{eq:brinkmann-conformal-profile}
 \widehat p(U)=L^{-4}p(u(U))-\frac{L''(U)}{L(U)}I.
\end{equation}
Here primes on $L$ denote $U$-derivatives.  Since
$\{u,U\}=-2L''/L$, equation
\eqref{eq:brinkmann-conformal-profile} becomes
\begin{equation*}
 2\widehat p(U)=L^{-4}\,2p(u(U))+\{u,U\}I,
\end{equation*}
the curvature form of \eqref{eq:grassmannian-schwarzian-chain}.  Passing to
trace-free parts gives the same law in the contact, Jacobi, and Brinkmann
descriptions.
\end{remark}

\begin{remark}[Tangential Jacobi fields]
The rank-two space of tangential Jacobi fields
\begin{equation*}
 \mathcal T_\gamma=\{(a+bu)k:a,b\in\mathbb R\}
\end{equation*}
is independent of the chosen affine parameter and gives the familiar quotient
$T_p\cN\cong\mathcal J_\gamma^{\mathrm{full}}/\mathcal T_\gamma$.  This
quotient removes affine origins and scales.  The larger projective atlas comes
from the trace of the sky Schwarzian.
\end{remark}

\subsection{Relation to the contact-scale jet}

The projective refinement complements the base $J^1\cS$ in
\Cref{thm:canonical-soldering}.  It supplies a class of longitudinal
parameters; a contact-scale jet selects the affine normalization and Reeb
representative used in metric-level soldering.  Forgetting that scale leaves
the projective structure along the ray, the conformal symplectic structure on
$\cC$, and the projectively weighted Weyl tidal profile
\eqref{eq:tracefree-schwarzian-law}.

The projective class supplies the longitudinal parameter geometry needed to
formulate the Hamilton--Jacobi constructions of
\cite{HollandSparlingIV,HollandSparlingV} over ray space.  Their variation over
the soldered bundle will be considered elsewhere.

\bibliographystyle{alpha}
\bibliography{pw6-soldered-penrose-limits-proof-corrections-2026-08-09}

\end{document}